\documentclass[11pt]{article}

\usepackage[T1]{fontenc}

\usepackage{lineno,hyperref}
\modulolinenumbers[5]

\usepackage{amsmath}
\usepackage{amssymb}
\usepackage{amsthm}
\usepackage{tikz}
\usepackage{pgfplots}
\usepackage{graphicx}
\usepackage{slashbox}
\usepackage{caption}
\usepackage{subcaption}
\usepackage{xcolor}
\usepackage{balance}
\allowdisplaybreaks  
\usepackage{cite}
\usepackage{algorithm}
\usepackage[noend]{algcompatible}
\usepackage{mathtools}
\usepackage{subcaption}
\usepackage{hyperref}
\usepackage{comment}
\usepackage{tabularx}
\usepackage{bbm}
\usepackage{thmtools}

\usepackage{hyperref}
\usepackage{pgfplots}
\usepackage{graphicx}
\usepackage{caption}
\usepackage{subcaption}
\usepackage{algorithm}
\usepackage[noend]{algcompatible}
\usepackage{multirow}
\usepackage{booktabs}
\usepackage{enumitem}
\usepackage{array}
\usepackage{authblk}
\usepackage{setspace}
\usepackage{fullpage,etoolbox}
\allowdisplaybreaks

\newcommand{\inp}{u}
\newcommand{\out}{y}
\newcommand{\sigmetric}{d}

\begin{document}

\newtheorem{theorem}{Theorem}
\newtheorem{definition}{Definition}
\newtheorem{lemma}{Lemma}
\newtheorem{proposition}{Proposition}

\title{Conformance Testing for Stochastic Cyber-Physical Systems}

\author[1]{Xin Qin}
\author[1]{Navid Hashemi}
\author[1]{Lars Lindemann}
\author[1]{Jyotirmoy V. Deshmukh}
\affil[1]{Thomas Lord Department of Computer Science, University of Southern California}

%
%
%
%
        
\maketitle

\begin{abstract}
Conformance is defined as a measure of distance between the behaviors of two dynamical systems. The notion of conformance can accelerate system design when models of varying fidelities are available on which analysis and control design can be done more efficiently. Ultimately, conformance can capture distance between design models and their real implementations and thus aid in robust system design. In this paper, we are interested in the conformance of stochastic dynamical systems. We argue that probabilistic reasoning over the distribution of distances between model trajectories is a good measure for \emph{stochastic conformance}. Additionally, we propose the \emph{non-conformance risk} to reason about the risk of stochastic systems not being conformant. We show that both notions have the desirable  \emph{transference} property, meaning that conformant systems satisfy similar system specifications, i.e., if the first model satisfies a desirable specification, the second model will satisfy (nearly) the same specification. Lastly, we propose how stochastic conformance and the non-conformance risk can be estimated from data using statistical tools such as conformal prediction. We present empirical evaluations of our method on an F-16 aircraft, an autonomous vehicle, a spacecraft, and Dubin's vehicle.

\end{abstract}


\section{Introduction}
\label{sec:introduction}

Cyber-physical systems (CPS) are usually designed using a model-based design (MBD) paradigm. Here, the designer models the physical parts and the operating environment of the system and  then designs the software used for perception, planning, and low-level control. Such closed-loop systems are then rigorously tested against various operating conditions, where the quality of the designed software is evaluated against {\em model properties} such as formal design specifications (or other kinds of quantitative objectives).  Examples of such property-based analysis techniques include requirement falsification \cite{bartocci2018specification,thibeault2021psy,akazaki2018falsification,deshmukh2019formal,qin2021automatic}, 
nondeterministic and statistical verification 
\cite{baier2008principles,wang2019statistical, legay2010statistical,legay2015statistical, agha2018survey,AlessandroAbate,KristinRozier,surrogateQ}, and risk analysis \cite{lindemann2022risk,akella2022sample}.

MBD is a fundamentally iterative process in which the  designer continuously modifies the software to tune performance or increase safety margins, or change plant models to perform design space exploration \cite{pimentel2016exploring}, e.g., using model abstraction or simplification \cite{schilders2008model,alur2000discrete,belta2017formal}, or to incorporate new data \cite{polydoros2017survey}. Any change to the system model, however, requires repeating the property-based analyses as many times as the number of system properties. The fundamental problem that we consider in this paper is that of  conformance  \cite{conformance_survey,deshmukh2015quantifying,abbas2014conformance,abbas2014formal,abbas2015test}. The notion
of conformance is defined w.r.t. the input-output behavior of a model. Typically, model inputs include exogenous disturbances or user-inputs to
the model, user-controllable design parameters, and initial operating  conditions. For a given input $\inp$, let $\out = S(\inp)$ 
denote the observable behavior of the model $S$. Furthermore, let $\sigmetric(\out_1,\out_2)$ be a metric defined over the space of the model behaviors. For deterministic models, two models $S_1$, $S_2$ are said to be
$\delta$-{\em conformant} if for all inputs $\inp$ it holds that $\sigmetric(\out_1,\out_2) < \delta$ where
$\out_1 = S_1(\inp)$ and  $\out_2 = S_2(\inp)$  \cite{deshmukh2015quantifying,abbas2014conformance,abbas2015test}. This notion of deterministic conformance is useful to reason about worst-case differences between models. However, most CPS applications use components that exhibit stochastic behavior; for example, sensors have measurement noise, actuators can have manufacturing variations, and most physical phenomena are inherently stochastic. The central question that this paper 
considers is: {\em What is the notion of conformance between two stochastic CPS models?} 

There are some challenges in comparing stochastic CPS models;
even if two models are repeatedly excited by the same input,
the pair of model behaviors that are observed may be different
for every such simulation. Thus, the observable behavior of a 
stochastic model is more accurately characterized by a distribution
over the space of trajectories. A possible way to compare two
stochastic models is to use measure-theoretic techniques to compare
the distance between the trajectory distributions. A number of 
divergence measures such as the f-divergences, e.g., the Kullback-Leibler divergence and the total
variation distance, or the Wasserstein distance  may look like candidate tools to compare the trajectory
distributions. However, we argue in this paper that a divergence 
is not the right notion to use to compare stochastic CPS models.
There can be two stochastic models whose output trajectories are
very close using any trajectory space metric, but the divergence
between their trajectory distributions can be infinite. On the
other hand, there can be two trajectory distributions with zero
divergence for which the distance between trajectories can be 
arbitrarily far apart.

This raises an interesting question: how do we then compare two stochastic models? In this paper, we argue that probabilistic bounds derived from the distribution of the distances between model trajectories (excited by the same input) gives us a general definition of conformance that has several advantages, as outlined below. We complement this probabilistic viewpoint further and capture the risk that the distribution of the distancs between model trajectories is large leveraging risk measures \cite{majumdar2020should}. 

First, we show that
two stochastic systems that are conformant under our definition inherit the
property of {\em transference} \cite{deshmukh2015quantifying}. In
simple terms, transference is the property that if the first model
has certain logical or quantitative properties, then the second
model also satisfies the same (or nearly same) properties. This
property brings several benefits. Consider the scenario where
probabilistic guarantees that a model has certain quantitative properties have been established after an extensive and large
number of simulations. Ordinarily, if there were any changes 
made to the model, establishing such probabilistic guarantees would
require repeating the extensive simulation-based procedure. However,
transference allows us to potentially sample from 
existing simulations for the first model and sample a small number
of simulations from the modified model to establish stochastic
conformance between the models, thereby allowing us to establish
probabilistic guarantees on the {\em second} model. We demonstrate
examples of such transference w.r.t. quantitative properties 
arising from quantitative semantics of temporal logic 
specifications and control-theoretic cost functions.


Next, we show how we can efficiently compute these probabilistic bounds
using the notion of {\em conformal prediction} \cite{tibshirani2019conformal,vovk2005algorithmic} from statistical learning theory. At a high-level, conformal prediction involves 
computing quantiles of the empirical distribution of 
non-conformity scores over a validation dataset to obtain prediction intervals at a given
confidence threshold. 


The contributions of this paper are summarized as follows:
\begin{itemize}
    \item We define \emph{stochastic conformance} as a probabilistic bound over the distribution of distances between model trajectories. We also define the \emph{non-conformance risk} to detect systems that are at risk of not being conformant.
    \item We show that both notions have the desirable transference property, meaning that conformant systems satisfy similar system specifications.
    \item We show how stochastic conformance and the non-conformance risk can be estimated  using statistical tools from risk theory and conformal prediction.
\end{itemize}

\section{Problem Statement and Preliminaries}
\label{sec:pf}

Consider the probability space $(\Omega,\mathcal{F},P)$  where $\Omega$ is
the sample space, $\mathcal{F}$ is a $\sigma$-algebra\footnote{A 
$\sigma$-algebra on a set $\Omega$ is a nonempty collection of subsets of 
$\Omega$ closed under complement, countable unions, and countable 
intersections.} of $\Omega$, and $P:\mathcal{F}\to[0,1]$ is a probability 
measure. In this paper, our goal is to quantify conformance of stochastic 
systems, i.e., systems whose inputs and outputs form a probability space 
with an appropriately defined measure. Let the two stochastic systems be 
denoted by $S_1$ and $S_2$. The inputs and outputs of stochastic systems
are {\em signals}, i.e., functions from a bounded interval
of positive reals known as the {\em time domain} $\mathbb{T} \subseteq 
\mathbb{R}_{\ge 0}$ to a metric space, e.g., the standard Euclidean metric.
Each stochastic system $S_i$ then describes an 
input-output relation $S_i:\mathcal{U}\times \Omega\to\mathcal{Y}$ where 
$\mathcal{U}$ and $\mathcal{Y}$ denote the sets of all input and output 
signals. We allow input signals\footnote{Probability spaces over signals
are defined by standard notions of cylinder sets \cite{baier2008principles}.}
to be stochastic, and we use the 
notation $U:\mathbb{T}\times \Omega \to \mathbb{R}^m$ to denote a stochastic 
input signals.\footnote{We will instead of the probability measure $P$, defined over $(\Omega,\mathcal{F})$, use more generally the notation $\text{Prob}$ to be independent of the underlying probability space that we induce, e.g., as a result of transformations via $U$.} Modeling stochastic systems this way provides great 
flexibility, and $S_i$ can e.g., describe the motion of stochastic hybrid 
systems, Markov chains, and stochastic difference equations. 

Assume now that we apply the input signal $U:\mathbb{T}\times \Omega \to \mathbb{R}^m$
to systems $S_1$ and $S_2$, and let the resulting output signals be denoted by $Y_1:\mathbb{T}\times \Omega \to \mathbb{R}^n$ and $Y_2:\mathbb{T}\times \Omega \to \mathbb{R}^n$, respectively. We assume that the functions $S_1$, $S_2$, and $U$ are measurable so that the output signals $Y_1$ and $Y_2$ are well-defined stochastic signals. One can hence think of $Y_1$ and $Y_2$ to be drawn from the distributions $\mathcal{D}_1$ and $\mathcal{D}_2$, respectively, which are functions of the probability space $(\Omega,\mathcal{F},P)$ as well as the functions $S_1$, $S_2$, and $U$. In this paper, we make no restricting assumptions on the functions $S_1$, $S_2$, and $U$, and consequently we make no assumptions on the distributions  $\mathcal{D}_1$ and $\mathcal{D}_2$.

\textbf{Informal Problem Statement.} Let $Y_1$ and $Y_2$ be stochastic output signals of the stochastic systems $S_1$ and $S_2$, respectively, under the stochastic input signal $U$. How can we quantify closeness of the stochastic systems $S_1$ and $S_2$ under $U$? To answer this question, we will explore different ways of defining system ``closeness'' of $Y_1$ and $Y_2$, and we will present algorithms to compute these stochastic notions of closeness. A subsequent problem that we consider is related to transference of properties from one system to another system. Particularly, given a signal temporal logic specification, can we infer guarantees about the satisfaction of the specification of one system from another system if the systems are close under a suitable definition of closeness?

 \subsection{Distance Metrics and Risk Measures} To define a general framework for quantifying closeness of stochastic systems, we will use i)  different signal metrics to capture the distance between individual realizations $y_1:=Y_1(\cdot,\omega)$ and $y_2:=Y_2(\cdot,\omega)$ of the stochastic signals $Y_1$ and $Y_2$  where $\omega\in\Omega$ is a single outcome, and ii) probabilistic reasoning and risk measures to capture stochastic conformance and non-conformance, respectively, under these signal metrics.

We first equip the set of output signals $\mathcal{Y}$ with a function $d:\mathcal{Y}\times\mathcal{Y}\to \mathbb{R}$ that quantifies distance between signals. A natural choice of $d$ is a signal metric that results in a metric space $(\mathcal{Y},d)$. We use general signal metrics such as the metric induced by the $L_p$ signal norm for $p\ge 1$. Particularly,  define
    \(d_p(y_1,y_2):=\Big(\int_{\mathbb{T}} \|y_1(t)-y_2(t)\|^p\mathrm{d}t\Big)^{1/p}\)
so that the $L_\infty$ norm can also be expressed as $d_\infty(y_1,y_2):=\sup_{t\in \mathbb{T}} \|y_1(t)-y_2(t)\|$.

It is now easy to see that a signal metric $d(Y_1,Y_2)$ evaluated over the stochastic signals $Y_1$ and $Y_2$ results in a distribution over distances between realizations of $Y_1$ and $Y_2$. To reason over properties of $d(Y_1,Y_2)$, we will use probabilistic reasoning but we will also consider risk measures \cite{majumdar2020should} as introduced next. 

A risk measure is a function $R:\mathfrak{F}(\Omega,\mathbb{R})\to \mathbb{R}$ that maps the set of real-valued random variables  $\mathfrak{F}(\Omega,\mathbb{R})$ to the real numbers. Typically, the input of $R$ indicates a cost. There exist various risk measures that capture different characteristic of the distribution of the cost random variable, such as the mean or the variance. However, we are particularly interested in tail risk measures that capture the right tail of the cost distribution, i.e., the potentially rare but costly outcomes.

In this paper, we particularly consider the value-at-risk $VaR_\beta$ and the conditional value-at-risk $CVaR_\beta$ at risk level $\beta\in (0,1)$.
The $VaR_\beta$ of a  random variable $Z:\Omega\to\mathbb{R}$ is defined as
\begin{align*}
VaR_\beta(Z):=\inf \{ \alpha \in \mathbb{R} |  \text{Prob}(Z\le\alpha) \ge \beta \}, 
\end{align*}
i.e., $VaR_\beta(Z)$ captures the $1-\beta$ quantile of the distribution of $Z$ from the right. Note that there is an obvious connection between value-at-risk and chance constraints, i.e., it holds that $\text{Prob}(Z\le \alpha)\ge \beta$ is equivalent to $VaR_\beta(Z)\le \alpha$. The $CVaR_\beta$ of $Z$, on the other hand, is defined as
\begin{align*}
CVaR_\beta(Z):=\inf_{\alpha \in \mathbb{R}}  \big(\alpha+(1-\beta)^{-1}E([Z-\alpha]^+)\big)
\end{align*} 
where $[Z-\alpha]^+:=\max(Z-\alpha,0)$ and $E(\cdot)$ indicates the expected value. When the function $\text{Prob}(Z\le\alpha)$ is continuous (in $\alpha$), it holds that $CVaR_\beta(Z)=E(Z|Z\ge VaR_\beta(Z))$, i.e., $CVaR_\beta(Z)$ is the expected value of $Z$ conditioned on the outcomes where $Z$ is greater or equal than  $VaR_\beta(Z)$. Finally, note that it holds that $VaR_\beta(Z)\le CVaR_\beta(Z)$, i.e., $CVaR_\beta$ is more risk sensitive.

For our risk transference results that we present later, we will require that $R$ is \emph{monotone}, \emph{positive homogeneous}, and \emph{subadditive}:
 \begin{itemize}[nosep,leftmargin=1em]
     \item For two random variables $Z,Z'\in \mathfrak{F}(\Omega,\mathbb{R})$, the risk measure $R$ is monotone if $Z(\omega) \leq Z'(\omega)$ for all $\omega\in\Omega$ implies that $R(Z) \le R(Z')$.
     \item For a random variable $Z\in \mathfrak{F}(\Omega,\mathbb{R})$, the risk measure $R$ is positive homogeneous if, for any constant $H\ge 0$, it holds that $R(HZ) = HR(Z)$.
     \item For two random variables $Z,Z'\in \mathfrak{F}(\Omega,\mathbb{R})$, the risk measure $R$ is subadditive if  $R(Z + Z') \le R(Z) + R(Z')$.
 \end{itemize} 
 We remark that  the $VaR_\beta$ and the $CVaR_\beta$ satisfies all three properties \cite{majumdar2020should}.


\subsection{System specifications} To express  specifications, we use Signal Temporal Logic (STL). Let $y:\mathbb{T}\to\mathbb{R}^n$ be a deterministic signal, e.g., a realization of the stochastic signal $Y$. The atomic elements of STL are predicates that are functions $\mu:\mathbb{R}^n\to\{\text{True},\text{False}\}$. For convenience, the predicate $\mu$ is often defined via a predicate function $h:\mathbb{R}^n\to\mathbb{R}$ as $\mu(y(t)):=\text{True}$ if $h(y(t))\ge 0$ and $\mu(y(t)):=\text{False}$ otherwise. The syntax of STL is recursively defined as 
\begin{align}\label{eq:full_STL}
\phi \; ::= \; \text{True} \; | \; \mu \; | \;  \neg \phi' \; | \; \phi' \wedge \phi'' \; | \; \phi'  U_I \phi'' 
\end{align}
where $\phi'$ and $\phi''$ are STL formulas. The Boolean operators $\neg$ and $\wedge$ encode negations (``not'') and conjunctions (``and''), respectively. The until operator $\phi' {U}_I \phi''$ encodes that $\phi'$ has to be true until $\phi''$ becomes true at some future time within the time interval $I\subseteq \mathbb{R}_{\ge 0}$. We  derive the operators for disjunction ($\phi' \vee \phi'':=\neg(\neg\phi' \wedge \neg\phi'')$), eventually ($F_I\phi:=\top U_I\phi$), and always ($G_I\phi:=\neg F_I\neg \phi$).

To determine if a signal $y$ satisfies an STL formula $\phi$ that is imposed at time $t$, we can define the semantics as a relation $\models$, i.e.,  $(y,t) \models\phi$ means that $\phi$ is satisfied. While the STL semantics are fairly standard \cite{maler2004monitoring}, we recall them in Appendix~\ref{app:STL}. Additionally, we can define robust semantics $\rho^{\phi}(y,t)\in\mathbb{R}$ that indicate how robustly the formula $\phi$ is satisfied or violated \cite{donze2,fainekos2009robustness}, see Appendix~\ref{app:STL}. Larger and positive values of $\rho^{\phi}(y,t)$ hence indicate that the specification is satisfied more robustly. Importantly, it holds that $(y,t)\models \phi$ if $\rho^\phi(y,t)>0$.


\section{Conformance for Stochastic Input-Output Systems}

Our goal is now to quantify closeness of  two stochastic systems $S_1$ and $S_2$ under the input $U$. We present our definitions for stochastic conformance and non-conformance risk upfront, and provide motivation for these afterwards. 
\begin{definition}\label{def:1}
Let $U:\mathbb{T}\times\Omega\to\mathbb{R}^m$ be a stochastic input signal, $S_1,S_2:\mathcal{U}\times\Omega\to\mathcal{Y}$ be  stochastic systems, and $Y_1,Y_2:\mathbb{T}\times\Omega\to\mathbb{R}^n$ be stochastic output signals with $Y_1:=S_1(U,\cdot)$ and $Y_2:=S_2(U,\cdot)$. Further, let $\epsilon\in\mathbb{R}$ be a conformance threshold, $\delta\in (0,1)$ be a failure probability, and $d:\mathcal{Y}\times\mathcal{Y}\to\mathbb{R}$ be a signal metric. Then, we say that the systems $S_1$ and $S_2$ under the input $U$ are $(\epsilon,\delta)$-conformant if
\begin{align}\label{eq:conformance}
    \text{Prob}(d(Y_1,Y_2)\le \epsilon)\ge 1-\delta.
\end{align}
Additionally, let $R:\mathfrak{F}(\Omega,\mathbb{R})\to \mathbb{R}$ be a risk measure and $r\in\mathbb{R}$ be a risk threshold. Then, $S_1$ and $S_2$ under the input $U$ are at risk of being $r$-non-conformant if
\begin{align}\label{eq:risk_conformance}
    R(d(Y_1,Y_2))> r.
\end{align}
\end{definition}

\noindent Eq.~\eqref{eq:conformance} is referred to as stochastic conformance and Eq.~\eqref{eq:risk_conformance} as non-conformance risk. Let us now motivate and discuss these two definitions. While the definition of conformance in equation \eqref{eq:conformance} appears natural at first sight, there are at least two competing ways of defining stochastic conformance. First, as $Y_1$ and $Y_2$ are distributions, it would be possible to define conformance as $D(Y_1,Y_2)$ where $D$ is a distance function that measures the difference between two distributions, such as a divergence (Kullback–Leibler or $f$-divergence). However, our definition provides an intuitive interpretation in the signal space where system specifications are typically defined, while it is usually difficult to provide such an interpretation for divergences between distributions. Additionally, the divergence between $Y_1$ and $Y_2$ may be unbounded (or zero) even when equation \eqref{eq:conformance} holds (does not hold). 

\begin{proposition}
There exist stochastic systems $S_1$ and $S_2$ and distance metrics $d$ where equation~\eqref{eq:conformance}: i) is satisfied for $\epsilon>0$ and $\delta = 0$, i.e., w.p. $1$, but where the divergence between
the systems is unbounded, and ii) is not satisfied for any given $\epsilon>0$ and $\delta\in(0,1)$, but where the divergence between the systems is zero.
\end{proposition}
\begin{proof}
Let us first prove i). For simplicity, consider systems $S_1$ and $S_2$ where the stochastic input and output signals
are defined over the time domain $\mathbb{T} := \{t_0,\ldots,t_T\}$. Further,
for all $t \in \mathbb{T}$ let $y_1(t) := 0$ and $y_2(t) := \epsilon$.
Clearly, equation~\eqref{eq:conformance} is satisfied, e.g., for $d_\infty$. The 
distributions $D_1$ and $D_2$ (joint distributions of
$Y_1(t)$ and $Y_2(t)$, respectively) are Dirac distributions
centered at $0$ and $\epsilon$, respectively. The Kullback–Leibler divergence between these two
distributions is $\infty$ \cite{gray2011entropy}.

Let us now prove ii). Let $\mathbb{T}$ consist of a single time point for simplicity so that $Y_1$ and $Y_2$ are random variables defined over a sample space $\mathbb{R}$. Let $D_1$ and $D_2$ be the
{\em same} uniform distribution over $[0,a]$. Clearly, the divergence between
$D_1$ and $D_1$ is zero. \textcolor{black}{We know  that the distribution of $Y:=Y_1-Y_2$ has support on $[-a,a]$, and that the probability density function of $Y$ is $p(y):=1/a-|y|/a^2$. We can now compute  $\text{Prob}(|y|\le \epsilon)=2\epsilon/a-\epsilon^2/a^2$.  Given $\epsilon>0$ and $\delta\in(0,1)$, we pick an $\bar{\delta}\in (0,1)$ such that $\bar{\delta}>\delta$. We then solve the quadratic equation $2\epsilon/a-\epsilon^2/a^2=1-\bar{\delta}$ subject to the constraint that $\epsilon\le a$. Consequently, we find that $a\ge \epsilon/(1-\bar{\delta})(1+\sqrt{\bar{\delta}})$ results in $\text{Prob}(|y|\le \epsilon)< 1-\delta$ so that \eqref{eq:conformance} is not satisfied.} 
\end{proof}




Another way of defining stochastic conformance was presented in \cite{wang2021probabilistic} where the authors consider a task-specific definition of 
stochastic conformance where satisfaction probabilities are required to be 
approximately equal. In other words, two stochastic systems are called $c$-approximately 
probabilistically conformant if $|\text{Prob}((Y_1,\tau)\models\phi)-\text{Prob}((Y_2,\tau)\models\phi)|\le c$. In this definition, it may happen that two systems are $c$-approximately probabilistically conformant for a small value of $c$, while the systems produce completely different behaviors and individual realizations $y_1$ and $y_2$ are vastly different. Additionally to not being task specific, our definition covers the risk of being $r$-non-conformant in equation \eqref{eq:risk_conformance}.

Finally, we would like to remark that the definition of conformance in equation \eqref{eq:conformance} is related to the definition of non-conformance risk in equation \eqref{eq:risk_conformance}. In fact, when the risk measure $R$ is the value-at-risk $VaR_\beta$, then we know that 
\begin{align*}
    VaR_\beta(d(Y_1,Y_2))> r \;\;\; \Leftrightarrow \;\;\; \text{Prob}(d(Y_1,Y_2) \le r)<\beta 
\end{align*}
since $VaR_\beta(d(Y_1,Y_2))\le r$ is equivalent to $\text{Prob}(d(Y_1,Y_2) \le r)\ge \beta$ according to Section \ref{sec:pf}. Consequently, if $\beta:=1-\delta$ and $r:=\epsilon$ then  $VaR_\beta(d(Y_1,Y_2))>r$ implies that the systems $S_1$ and $S_2$ under $U$ are not $(\epsilon,\delta)$-conformant.

The notion of conformance in Definition \ref{def:1} is useful when
the input $U$ describes internal inputs such as system parameters (an 
unknown mass), exogeneous disturbances from known sources, or initial
system conditions. In other words, the distribution $U$ is known,
making $U$ a {\em known unknown}.
However, in case of external inputs that could be manipulated (e.g.
user inputs that represent rare malicious attacks), the input 
$U$ may be unknown, making $U$ an {\em unknown unknown}.
We therefore provide an alternative definition of conformance.
\begin{definition}\label{def:2}
Let $U\in \mathcal{U}$ be an unknown deterministic input signal,  $S_1,S_2:\mathcal{U}\times\Omega\to\mathcal{Y}$ be stochastic systems, and $Y_1,Y_2:\mathbb{T}\times\Omega\to\mathbb{R}$ be stochastic output signals with $Y_1:=S_1(U,\cdot)$ and $Y_2:=S_2(U,\cdot)$. Further, let $\epsilon\in\mathbb{R}$ be a conformance threshold, $\delta\in (0,1)$ be a failure probability, and $d:\mathcal{Y}\times\mathcal{Y}\to\mathbb{R}$ be a signal metric. Then, we say that the systems $S_1$ and $S_2$ are $(\epsilon,\delta)$-conformant if
\begin{align}\label{eq:conformance_}
    \text{Prob}\Big(\sup_{U\in\mathcal{U}} d(Y_1,Y_2)\le \epsilon\Big)\ge 1-\delta.
\end{align}
Additionally, let $R:\mathfrak{F}(\Omega,\mathbb{R})\to \mathbb{R}$ be a risk measure and $r\in\mathbb{R}$ be a risk threshold. Then, we say that the systems $S_1$ and $S_2$ under the input $U$ are at risk of being $r$-non-conformant if
\begin{align}\label{eq:risk_conformance_}
    R\Big(\sup_{U\in\mathcal{U}}d(Y_1,Y_2)\Big)> r.
\end{align}
\end{definition}

Based on this definition, note that it will be inherently more difficult to verify Definition \ref{def:2} compared to Definition \ref{def:1} due to the $\sup$-operator.

\section{Transference of System Properties under Conformance}
We expect two systems $S_1$ and $S_2$ that are $(\epsilon,\delta)$-conformant in the sense of Definitions \ref{def:1} and \ref{def:2} to have similar behaviors with respect to satisfying a given system specification. Therefore, we will define the notion of transference with respect to a performance function $C:\mathcal{Y}\to\mathbb{R}$ that measures how well a signal $y\in \mathcal{Y}$ satisfies this system specification. Towards capturing similarity between $S_1$ and $S_2$ with respect to $C$, the signal metric $d$ has to be chosen carefully. 
\begin{definition}\label{def:3}
Let $d:\mathcal{Y}\times\mathcal{Y}$ be a signal metric and $C:\mathcal{Y}\to\mathbb{R}$ be a performance function. Then, we say that $C$ is Hölder continuous w.r.t. $d$ if there exists constants $H,\gamma>0$ such that, for any two signals $y_1,y_1:\mathbb{T}\to\mathbb{R}^n$, it holds that
\begin{align}\label{eq:continuity}
    |C(y_1)-C(y_2)|\le H d(y_1,y_2)^\gamma
\end{align}
\end{definition}

A specific example of the performance function $C$ is the robust semantics $\rho^\phi$ of an STL specification $\phi$. In fact, the robust semantics $\rho^\phi$ are Hölder continuous w.r.t. the sup-norm $d_\infty$ for constants $H=1$ and $\gamma=1$ \cite[Lemma 2]{cleaveland2022risk}. For the convenience of the reader, we state the proof of \cite[Lemma 2]{cleaveland2022risk} with the notation used in this paper in Appendix \ref{sec:hölder}. The robust semantics are also Hölder continuous w.r.t. the Skorokhod metric, see Appendix \ref{sec:hölder2}. A commonly used performance function in control is $C(y) = \int_{0}^{T} y(t)^\top y(t)dt$, and we note that this choice of $C$ is Hölder continuous w.r.t. $d_1$ as shown in Appendix \ref{sec:hölder3}. Finally, note that the Hölder continuity condition in equation \eqref{eq:continuity}  implies that, for any constants $c,\epsilon\in \mathbb{R}$, it holds that
\begin{align}\label{eq:sensible} 
    \big(C(y_1)\ge c \;\;\wedge\;\; d(y_1,y_2)\le \epsilon\big) \;\;\; \Rightarrow \;\;\; C(y_2)\ge c-H \epsilon^\gamma.
\end{align}

\subsection{Transference under stochastic conformance} With the definition of $C$ being Hölder continuous w.r.t. $d$, we can now derive a stochastic transference result under stochastic conformance as per Definition \ref{def:1}.
 
\begin{theorem}\label{thm:1}
Let the premises in Definitions \ref{def:1} and \ref{def:3} hold. Further, let the systems $S_1$ and $S_2$ under the input $U$ be $(\epsilon,\delta)$-conformant so that equation \eqref{eq:conformance} holds and let $C$ be Hölder continuous w.r.t. $d$ so that equation \eqref{eq:continuity} holds. Then, it holds that
\begin{align*}
    \text{Prob}\big(C(Y_1)\ge c\big)\ge 1-\bar{\delta} \;\;\; \Rightarrow \;\;\;  \text{Prob}\big(C(Y_2)\ge c-H \epsilon^\gamma\big)\ge 1-\delta-\bar{\delta}.
\end{align*}
\end{theorem}
\begin{proof}
By assumption, it holds that $\text{Prob}(d(Y_1,Y_1)\le \epsilon)\ge 1-\delta$ and $\text{Prob}(C(Y_1)\ge c)\ge 1-\bar{\delta}$ so that we know that $\text{Prob}(d(Y_1,Y_1)> \epsilon)\le \delta$ and $\text{Prob}(C(Y_1)< c)\le \bar{\delta}$. We can now apply the union bound over these two events so that
\begin{align*}
    \text{Prob}\big(d(Y_1,Y_1)> \epsilon \; \vee \; C(Y_1)< c \big)\le \delta+\bar{\delta}.
\end{align*}
From here, we can simply see that 
\begin{align*}
    \text{Prob}\big(d(Y_1,Y_1)\le \epsilon \; \wedge \; C(Y_1)\ge c \big)\ge 1-\delta-\bar{\delta}.
\end{align*}
Since $C$ is Hölder continuous w.r.t. $d$, which implies that equation \eqref{eq:sensible} holds, it is  easy to conclude that $\text{Prob}(C(Y_2)\ge c-H\epsilon^\gamma)\ge 1-\delta-\bar{\delta}$.
\end{proof}

Theorem \ref{thm:1} tells us that i) $(\epsilon,\delta)$-conformance of systems $S_1$ and $S_2$ under $U$, and ii) Hölder continuity of the performance function $C$ w.r.t. the metric $d$ enables us to derive a probabilistic lower bound for the performance of system $S_2$ w.r.t. $C$ from the performance of system $S_1$. 

We can derive a transference result similar to Theorem \ref{thm:1} when we assume that the systems $S_1$ and $S_2$ are $(\epsilon,\delta)$-conformant in the sense of Definition \ref{def:2} instead of Definition \ref{def:1}.

 \begin{theorem}\label{thm:1_}
Let the premises in Definitions \ref{def:2} and \ref{def:3} hold. Further, let the systems $S_1$ and $S_2$ be $(\epsilon,\delta)$-conformant so that equation \eqref{eq:conformance_} holds and let $C$ be Hölder continuous w.r.t. $d$ so that equation \eqref{eq:continuity} holds. Then, it holds that
 \begin{align*}
    \text{Prob}\big(\inf_{U\in\mathcal{U}}C(Y_1)\ge c\big)\ge 1-\bar{\delta} \;\; \Rightarrow \;\;   \text{Prob}\big(\inf_{U\in\mathcal{U}}C(Y_2)\ge c-H \epsilon^\gamma\big)\ge 1-\delta-\bar{\delta}
\end{align*}
\end{theorem}

\begin{proof}
By assumption, it holds that $\text{Prob}(\sup_{U\in\mathcal{U}}d(Y_1,Y_1)\le \epsilon)\ge 1-\delta$ and $\text{Prob}(\inf_{U\in\mathcal{U}}C(Y_1)\ge c)\ge 1-\bar{\delta}$ so that we know that $\text{Prob}(\sup_{U\in\mathcal{U}}d(Y_1,Y_1)> \epsilon)\le \delta$ and $\text{Prob}(\inf_{U\in\mathcal{U}}C(Y_1)< c)\le \bar{\delta}$. We can now apply the union bound over these two events so that
\begin{align*}
    \text{Prob}\big(\sup_{U\in\mathcal{U}}d(Y_1,Y_1)> \epsilon \; \vee \; \inf_{U\in\mathcal{U}}C(Y_1)< c \big)\le \delta+\bar{\delta}.
\end{align*}
From here, we can simply see that 
\begin{align*}
    \text{Prob}\big(\sup_{U\in\mathcal{U}}d(Y_1,Y_1)\le \epsilon \; \wedge \; \inf_{U\in\mathcal{U}}C(Y_1)\ge c \big)\ge 1-\delta-\bar{\delta}.
\end{align*}
This equation tells us that, for each $U\in\mathcal{U}$, we have 
\begin{align*}
    \text{Prob}\big(d(Y_1,Y_1)\le \epsilon \; \wedge \; C(Y_1)\ge c \big)\ge 1-\delta-\bar{\delta}.
\end{align*}
Since $C$ is Hölder continuous w.r.t. $d$, we know that equation \eqref{eq:continuity} holds for each $U\in\mathcal{U}$. Consequently, we can conclude that $\text{Prob}(\inf_{U\in\mathcal{U}}C(Y_2)\ge c-H \epsilon^\gamma)\ge 1-\delta-\bar{\delta}$ .
\end{proof}

\subsection{Transference under non-conformance risk} 
On the other hand, by considering the notion of $r$-non-conformance risk, we expect that two systems $S_1$ and $S_2$ that are not at risk of being $r$-non-conformant have a similar risk of violating a specification. Here, we define the risk of violating a specifications by following ideas from \cite{lindemann2022risk} as $R(-C(Y_1))$ and $R(-C(Y_2))$.
\begin{theorem}\label{thm:2}
Let the premises in Definitions \ref{def:1} and \ref{def:3} hold. Further, let the systems $S_1$ and $S_2$ under the input $U$  not be at risk of being $r$-non-conformant  so that equation \eqref{eq:risk_conformance} does not hold (i.e., $R(d(Y_1,Y_2))\le r$) and let $C$ be Hölder continuous w.r.t. $d$ with $\gamma=1$ so that equation \eqref{eq:continuity} holds. If the risk measure $R$ is monotone, positive homogeneous, and subadditive, it holds that
\begin{align*}
  R(-C(Y_2))\le R(-C(Y_1))+Hr.
\end{align*}
\end{theorem}
\begin{proof}
We can derive the following chain of inequalities
\begin{align*}
    R(-C(Y_2)) &\overset{(a)}{\le} R(-C(Y_1)+Hd(Y_1,Y_1)) \\ & \overset{(b)} {\le} 
    R(-C(Y_1))+R(Hd(Y_1,Y_1))\\
    & \overset{(c)}{=}R(-C(Y_1))+HR(d(Y_1,Y_1)) \\
    &\overset{(d)}{\le} R(-C(Y_1))+Hr
\end{align*}
where $(a)$ follows since $C$ is Hölder continuous w.r.t. $d$ and since $R$ is monotone, $(b)$ follows since $R$ is subadditive, and $(c)$ follows since $R$ is positive homogeneous, while the inequality $(d)$ follows since  $S_1$ and $S_2$ under $U$ are not at risk of being $r$-non-conformant, i.e., $R(d(Y_1,Y_2))\le r$.
\end{proof}

This result implies that the risk of system $S_2$ w.r.t. the performance function $C$ is upper bounded by the risk of system $S_1$ w.r.t. $C$ if the systems $S_1$ and $S_2$ are not at risk of being $r$-non-conformant. We remark that a similar result  appeared in our prior work \cite{cleaveland2022risk}. Here, we present these results in the more general context of conformance and extend the result as we use general performance functions $C$, which additionally requires $R$ to be positive homogeneous. Additionally, we derive a transference result similar to Theorem \ref{thm:2} when we assume that the systems $S_1$ and $S_2$ are not at risk of being $r$-non-conformant in the sense of Definition \ref{def:2} instead of Definition \ref{def:1}. 


\begin{theorem}\label{thm:2_}
Let the premises in Definitions \ref{def:2} and \ref{def:3} hold. Further, let the systems $S_1$ and $S_2$ not be at risk of being $r$-non-conformant  so that equation \eqref{eq:risk_conformance_} does not hold (i.e., $R(\sup_{U\in\mathcal{U}}d(Y_1,Y_2))\le r$) and let $C$ be Hölder continuous w.r.t. $d$ with $\gamma=1$ so that equation \eqref{eq:continuity} holds. If the risk measure $R$ is monotone, positive homogeneous, and subadditive, it holds that
\begin{align*}
  R(-\inf_{U\in\mathcal{U}}C(Y_2))\le R(-\inf_{U\in\mathcal{U}}C(Y_1))+Hr.
\end{align*}
\end{theorem}
\begin{proof}
We can derive the following chain of inequalities
\begin{align*}
    R(-\inf_{U\in\mathcal{U}}C(Y_2)) &\overset{(a)}{\le} R(-\inf_{U\in\mathcal{U}}C(Y_1)+H\sup_{U\in\mathcal{U}}d(Y_1,Y_1))\\ &\overset{(b)}{\le} R(-\inf_{U\in\mathcal{U}}C(Y_1))+R(H\sup_{U\in\mathcal{U}}d(Y_1,Y_1))\\
    & \overset{(c)}{=}R(-\inf_{U\in\mathcal{U}}C(Y_1))+HR(\sup_{U\in\mathcal{U}}d(Y_1,Y_1))\\
    &\overset{(d)}{\le} R(-\inf_{U\in\mathcal{U}}C(Y_1))+Hr
\end{align*}
where $(a)$ follows since $-\inf_{U\in\mathcal{U}}C(Y_2)=\sup_{U\in\mathcal{U}}-C(Y_2)$, since $C$ is Hölder continuous w.r.t. $d$, and since $R$ is monotone, $(b)$ follows since $R$ is subadditive, and $(c)$ follows since $R$ is positive homogeneous. The inequality $(d)$ follows since  $S_1$ and $S_2$ are not at risk of being $r$-non-conformant, i.e., $\sup_{U\in\mathcal{U}}R(d(Y_1,Y_2))\le r$.
\end{proof}

\section{Statistical Estimation of Stochastic Conformance}
\label{sec:statest}
We  propose algorithms to compute stochastic conformance and the non-conformance risk. In practice, note that one will be limited to discrete-time stochastic systems to apply these algorithms.

\subsection{Estimating stochastic conformance} To estimate stochastic conformance, we use conformal prediction which is a statistical tool introduced  in \cite{vovk2005algorithmic,shafer2008tutorial} to obtain valid uncertainty regions for complex prediction models without making assumptions on the underlying distribution or the prediction model \cite{angelopoulos2021gentle,zeni2020conformal,lei2018distribution,tibshirani2019conformal,cauchois2020robust}.  Let $Z,Z^{(1)},\hdots,Z^{(k)}$ be $k+1$ independent and identically distributed random variables modeling a quantity
known as the \emph{nonconformity score}. Our goal is to obtain an uncertainty region for $Z$ based on $Z^{(1)},\hdots,Z^{(k)}$, i.e., the random variable $Z$ should be contained within the uncertainty region with high probability.  Formally, given a failure probability $\delta\in (0,1)$, we want to construct a valid uncertainty region over $Z$ (defined in terms of a value $\bar{Z}$) that depends on $Z^{(1)},\hdots,Z^{(k)}$ such that 
 \(   \text{Prob}(Z\le \bar{Z})\ge 1-\delta \).

By a surprisingly simple quantile argument, see \cite[Lemma 1]{tibshirani2019conformal}, one can obtain $\bar{Z}$ to be the $(1-\delta)$th quantile of the empirical distribution of the values $Z^{(1)},\hdots,Z^{(k)}$ and $\infty$. By assuming that $Z^{(1)},\hdots,Z^{(k)}$ are sorted in non-decreasing order, and by adding $Z^{(k+1)}:=\infty$, we can equivalently obtain $\bar{Z}:=Z^{(p)}$ where $p:=\lceil (k+1)(1-\delta)\rceil$ with $\lceil\cdot\rceil$ being the ceiling function.

We can now use conformal prediction to estimate stochastic conformance as defined in Definition \ref{def:1} by setting $Z:=d(Y_1,Y_2)$. We therefore assume that we have access to a calibration dataset $D_\text{cal}$ that consists of realizations $y_1^{(i)}$ and $y_2^{(i)}$ from the stochastic signals $Y_1\sim \mathcal{D}_1$ and $Y_2\sim\mathcal{D}_2$, respectively. 
\begin{theorem}\label{thm:5}
 Let the premises of Definition \ref{def:1} hold and $D_\text{cal}$ be a calibration dataset with datapoints $(y_1^{(i)},y_2^{(i)})$ drawn from $\mathcal{D}_1\times\mathcal{D}_2$. Further, define $Z^{(i)}:=d(y_1^{(i)},y_2^{(i)})$ for all $i\in\{1,\hdots,|D_\text{cal}|\}$ and $Z^{(|D_\text{cal}|+1)}:=\infty$, and assume that the $Z^{(i)}$ are sorted in non-decreasing order. Then, it holds that  $\text{Prob}(d(Y_1,Y_2)\le \bar{Z})\ge 1-\delta$ with $\bar{Z}$ defined as $\bar{Z}:=Z^{(p)}$ where $p:=\big\lceil (|D_\text{cal}|+1)(1-\delta)\big\rceil$. Thus, the systems $S_1$ and $S_2$ under the input $U$ are $(\epsilon,\delta)$-conformant if $\bar{Z}\le \epsilon$.
\end{theorem}

We see that checking stochastic conformance as defined in Definition \ref{def:1} is computationally simple when we have a calibration dataset $D_\text{cal}$. Checking stochastic conformance as defined in Definition~\ref{def:2}, however, is  more difficult due to the existence of the $\sup$-operator. To compute this notion of conformance, we make two assumptions: i) the set $\mathcal{U}$ is compact, and ii) for every realization $\omega\in\Omega$, the function $d(Y_1(\cdot,\omega),Y_2(\cdot,\omega))$ is Lipschitz continuous with Lipschitz constant $L$. While knowledge of the Lipschitz constant $L$ would presume knowledge about the closeness of the systems $S_1$ and $S_2$, it would only provide a conservative over-approximation. We will, however, not need to know the Lipschitz constant $L$ and estimate $L$ instead along with probabilistic guarantees.

Our approach is summarized in Algorithms \ref{alg:overview} and \ref{alg:overview_}. Algorithm \ref{alg:overview} computes $\bar{Z}$ such that $\text{Prob}(\sup_{U\in\mathcal{U}}d(Y_1,Y_2)\le \bar{Z}+L\kappa)\ge 1-\delta$ when $L$ is known and where $\kappa$ is a gridding parameter, while Algorithm \ref{alg:overview_} estimates the Lipschitz constant. We present a description of these algorithms upfront and state their theoretical guarantees afterwards.

In line 1 of Algorithm \ref{alg:overview}, we  construct a $\kappa$-net $\bar{\mathcal{U}}$ of $\mathcal{U}$, i.e., we construct a finite set $\bar{\mathcal{U}}$  so that for each $U\in\mathcal{U}$ there exists a $\bar{U}\in \bar{\mathcal{U}}$ such that $\bar{d}(U,\bar{U})\le \kappa$ where $\bar{d}:\mathcal{U}\times\mathcal{U}\to\mathbb{R}$ is a  metric. For this purpose, simple gridding strategies can be used as long as the set $\mathcal{U}$ has a convenient representation. Alternatively, randomized algorithms can be used that sample from $\mathcal{U}$ \cite{vershynin2018high}. In lines 2-4, we apply Theorem \ref{thm:5} for each element $\bar{U}\in\bar{\mathcal{U}}$. Therefore, we obtain realizations $(y_1^{(i)},y_2^{(i)})$ from $\mathcal{D}_1\times\mathcal{D}_2$ under $\bar{U}$ (line 3). We then compute $\bar{Z}_{\bar{U}}$ so that $\text{Prob}(d(Y_1(\bar{U},\cdot),Y_2(\bar{U},\cdot))\le \bar{Z}_{\bar{U}})\ge 1-\delta$ (line 4). Finally, we set $\bar{Z}:=\max_{\bar{U}\in\mathcal{\bar{U}}}\bar{Z}_{\bar{U}}$ (line 5).

In Algorithm \ref{alg:overview_}, we compute $\bar{L}$ such that $\text{Prob}(L\le \bar{L})\ge1-\delta_L$. We uniformly sample control inputs $(U',U'')$ (line 2), obtain realizations $(y_1',y_2')$ from $\mathcal{D}_1\times\mathcal{D}_2$  under $U'$ and realizations $({y}_1'',{y}_2'')$ from $\mathcal{D}_1\times\mathcal{D}_2$ under $U''$ (line 3), and compute the non-conformity score $L^{(i)}$ (line 4). In line 5, we obtain an estimate $\bar{L}$ of the Lipschitz constant $L$ that holds with a probability of $1-\delta_L$ over the randomness introduced in Algorithm \ref{alg:overview}. 
\begin{theorem}\label{th:supU}
 Let the premises of Definition \ref{def:2} hold. If the Lipschitz constant $L$ of $d(Y_1(\cdot,\omega),Y_2(\cdot,\omega))$ is known uniformly over $\omega\in\Omega$, then, for a gridding parameter $\kappa>0$, the output $\bar{Z}$ of Algorithm \ref{alg:overview} ensures that
  \begin{align*}
     \text{Prob}(\sup_{U\in\mathcal{U}}d(Y_1,Y_2)\le \bar{Z}+L\kappa)\ge 1-\delta
 \end{align*}
Thus, the systems $S_1$ and $S_2$ are $(\epsilon,\delta)$-conformant if $\bar{Z}+L\kappa\le \epsilon$. Otherwise, let $\delta_L\in(0,1)$ be a failure probability, then the output $\bar{L}$ of Algorithm \ref{alg:overview_} ensures that
 \begin{align*}
     \text{Prob}(\sup_{U\in\mathcal{U}}d(Y_1,Y_2)\le \bar{Z}+\bar{L}\kappa)\ge 1-\delta-\delta_L
 \end{align*}
 where $\text{Prob}$ is defined over the randomness introduced in Algorithm \ref{alg:overview_}. 
\end{theorem}

\begin{algorithm}[t]
    \centering
    \begin{algorithmic}[1]
        \Statex \textbf{Input: } Failure probability $\delta\in(0,1)$ and grid size $\kappa>0$
        \Statex \textbf{Output: } $\bar{Z}$ such that $\text{Prob}(\sup_{U\in\mathcal{U}}d(Y_1,Y_2)\le \bar{Z}+L\kappa)\ge 1-\delta$
        \State Construct $\kappa$-net $\bar{\mathcal{U}}$ of $\mathcal{U}$
        \FOR{$\bar{U}\in\bar{\mathcal{U}}$}
            \State Obtain calibration set $D_\text{cal}^{\bar{U}}$ consisting of realizations $(y_1^{(i)},y_2^{(i)})$ under $\bar{U}$
            \State Compute $\bar{Z}_{\bar{U}}:=Z^{(p)}$ by applying Theorem \ref{thm:5} but instead using dataset $D_\text{cal}^{\bar{U}}$
        \ENDFOR
        \State $\bar{Z}:=\max_{\bar{U}\in\mathcal{\bar{U}}}\bar{Z}_{\bar{U}}$
    \end{algorithmic}
    \caption{Conformance Estimation as per Definition \ref{def:2}}
    \label{alg:overview}
\end{algorithm}

\begin{algorithm}[t]
    \centering
    \begin{algorithmic}[1]
        \Statex \textbf{Input: } Failure probabilities $\delta_L\in(0,1)$, grid size $\kappa>0$, calibration size $K_L>0$
        \Statex \textbf{Output: } $\bar{L}$ such that $\text{Prob}(\sup_{U\in\mathcal{U}}d(Y_1,Y_2)\le \bar{Z}+\bar{L}\kappa)\ge 1-\delta-\delta_L$
        \FOR{$i$ from $1$ to $K_L$}
        \State Sample $(U',U'')$ uniformly from $\mathcal{U}\times \mathcal{U}$
        \State Obtain realizations $(y_1',y_2')$ under $U'$ and $(y_1'',y_2'')$ under $U''$
        \State Compute $L^{(i)}:={|d(y_1',y_2')-d(y_1'',y_2'')|}/{\bar{d}(U',U'')}$
        \ENDFOR
        \State Compute $\bar{L} := L^{(p)}$ where $p:=\big\lceil (K_L+1)(1-\delta'')\big\rceil$
    \end{algorithmic}
    \caption{Lipschitz Constant Estimation of $L$}
    \label{alg:overview_}
\end{algorithm}

 \begin{proof}
From line 4 of Algorithm \ref{alg:overview} we know that $\text{Prob}(d(Y_1(\bar{U},\cdot),Y_2(\bar{U},\cdot))\le \bar{Z}_{\bar{U}})\ge 1-\delta$ for each $\bar{U}\in\mathcal{U}$. Due to Lipschitz continuity, we can conclude that for each $U\in\mathcal{U}$ that is such that $\bar{d}(U,\bar{U})\le \kappa$ it holds that
\begin{align*}
    \text{Prob}(d(Y_1,Y_2)\le \bar{Z}_{\bar{U}}+{L}\kappa)\ge 1-\delta.
\end{align*} 
Since $\bar{\mathcal{U}}$ is a $\kappa$-net of $\mathcal{U}$, it follows that $\text{Prob}(\sup_{U\in\mathcal{U}}d(Y_1,Y_2)\le \bar{Z}+L\kappa)\ge 1-\delta$. 

For the second part of the proof, note that from line 5 of Algorithm \ref{alg:overview_} we know that $\text{Prob}(L\le \bar{L})\ge 1-\delta_L$. We can now union bound over this event and $\text{Prob}(d(Y_1(\bar{U},\cdot),Y_2(\bar{U},\cdot))\le \bar{Z}_{\bar{U}})\ge 1-\delta$  so that 
\begin{align*}
    \text{Prob}(d(Y_1(\bar{U},\cdot),Y_2(\bar{U},\cdot))\le \bar{Z}_{\bar{U}} \;\; \wedge \;\; L\le \bar{L})\ge 1-\delta-\delta_L.
\end{align*}
The rest of the proof follows as in the first part.
\end{proof}

\subsection{Estimating non-conformance risk}  We next briefly summarize how to estimate the value-at-risk and the conditional value-at-risk following standard results such as from \cite{lindemann2022risk,massart1990tight} and \cite{wang2010deviation}, respectively.
\begin{proposition}
     \sloppy Let the premises of Definition \ref{def:1} hold and $D_\text{cal}$ be a calibration dataset with datapoints $(y_1^{(i)},y_2^{(i)})$ drawn from $\mathcal{D}_1\times\mathcal{D}_2$. Let $\beta\in(0,1)$ be a risk level and $\gamma\in(0,1)$ be a failure threshold. Define $Z^{(i)}:=d(y_1^{(i)},y_2^{(i)})$ for each $i\in\{1,\hdots,|D_\text{cal}|\}$ and assume that $\text{Prob}(Z\le\alpha)$ is continuous in $\alpha$. Then, 
	\begin{align*}
	\text{Prob}\big(\underline{VaR}_\beta\le VaR_\beta(d(Y_1,Y_2))\le \overline{VaR}_\beta\big)\ge 1-\gamma.
	\end{align*}
	where we have $\overline{VaR}_\beta:=\inf\left\{\alpha\in {\mathbb{R}}\,\middle|\,\widehat{\text{Prob}}(Z\le \alpha)-\sqrt{\frac{\ln(2/\gamma)}{2|D_\text{cal}|}}\ge \beta\right\}$ and $\underline{VaR}_\beta:=\inf\left\{\alpha\in {\mathbb{R}}\,\middle|\,\widehat{\text{Prob}}(Z\le \alpha)+\sqrt{\frac{\ln(2/\gamma)}{2|D_\text{cal}|}}\ge \beta\right\}$ with the empirical cumulative distribution function $\widehat{\text{Prob}}(Z\le \alpha):=\frac{1}{|D_\text{cal}|}\sum_{i=1}^{|D_\text{cal}|} \mathbb{I}(Z^{(i)}\le \alpha)$ and the indicator function $\mathbb{I}$.
\end{proposition}

For estimating  $CVaR_\beta(Z)$, we assume that the random variable $d(Y_1,Y_2)$ has bounded support, i.e., that $\text{Prob}(d(Y_1,Y_2)\in[a,b])=1$. Note that $d(Y_1,Y_2)$ is usually bounded from below by $a:=0$ if $d$ is a metric. To obtain an upper bound, we assume that the distance function saturated at $b$, e.g., by clipping values larger than $b$ to $b$. In practice, this means that realizations that are far apart already have a large distance and are capped to $b$.
\begin{proposition}
     Let the premises of Definition \ref{def:1} hold and $D_\text{cal}$ be a calibration dataset with datapoints $(y_1^{(i)},y_2^{(i)})$ drawn from $\mathcal{D}_1\times\mathcal{D}_2$.  Let $\beta\in(0,1)$ be a risk level and $\gamma\in(0,1)$ be a failure threshold. Define $Z^{(i)}:=d(y_1^{(i)},y_2^{(i)})$ for each $i\in\{1,\hdots,|D_\text{cal}|\}$ and assume that $\text{Prob}(d(Y_1,Y_2)\in[a,b])=1$. Then, it holds that
	\begin{align*}
	\text{Prob}\big(\underline{CVaR}_\beta\le CVaR_\beta(d(Y_1,Y_2))\le \overline{CVaR}_\beta\big)\ge 1-\gamma.
	\end{align*}
	where $\overline{CVaR}_\beta:=\widehat{CVaR}_\beta+\sqrt{\frac{5\ln(3/\gamma)}{|D_\text{cal}|(1-\beta)}}(b-a)$ and $\underline{CVaR}_\beta:=\widehat{CVaR}_\beta-\sqrt{\frac{11\ln(3/\gamma)}{|D_\text{cal}|(1-\beta)}}(b-a)$ where the empirical estimate of $CVaR_\beta(Z)$ is $\widehat{CVaR}_\beta:=\inf_{\alpha\in{\mathbb{R}}}\big(\alpha+(|D_\text{cal}|(1-\beta))^{-1}\sum_{i=1}^{|D_\text{cal}|}[Z^i-\alpha]^+\big)$.
\end{proposition}

As a consequence of these two lemmas, we know that with a probability of $1-\gamma$ the systems $S_1$ and $S_2$ under the input $U$ are at risk of not being conformant if $\underline{VaR}_\beta\ge \alpha$ or $\underline{CVaR}_\beta\ge \alpha$ based on the risk measure of choice.

\section{Case Studies}
\label{sec:simulations}


%
%
%
\begin{figure*}[t]
\centering
\subfloat[Targets in Dubin's car]{
\includegraphics[width=0.12\textwidth]{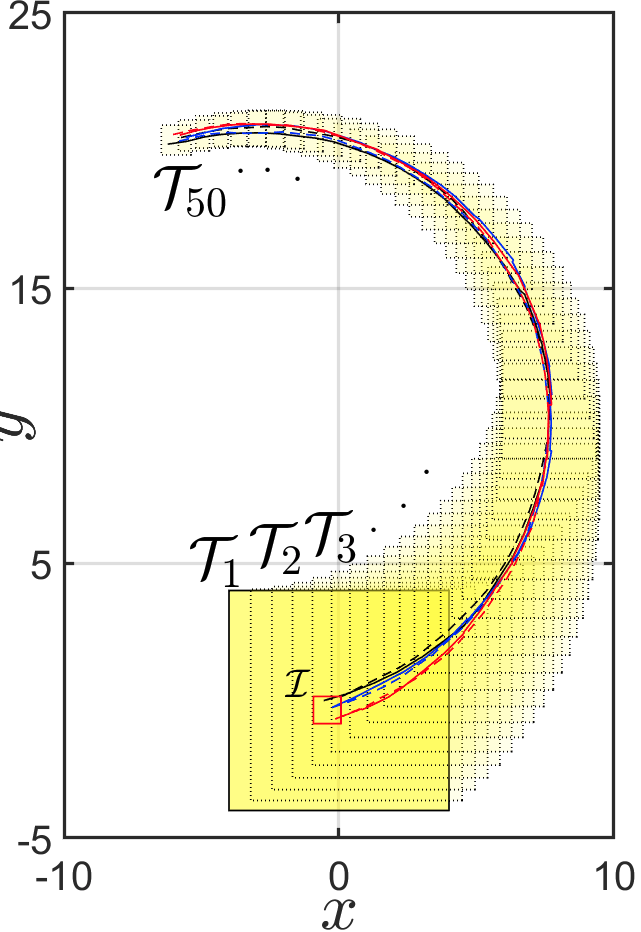}\label{subfig:Dubin}} \quad
\subfloat[CARLA: Cross-track error signals for $S_1,S_2$]{
\includegraphics[width=0.2489\linewidth, trim={0.5cm 0.15cm 0.5cm 0.15cm}, clip]{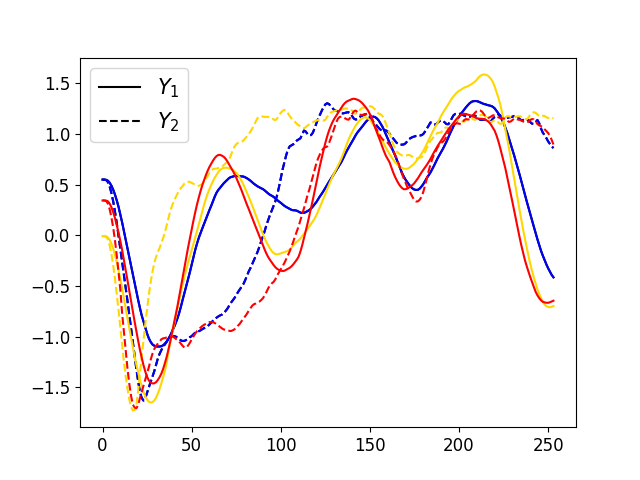}\label{carla_traj}} \quad
\subfloat[F16: altitude signals for $S_1,S_2$]{
\includegraphics[width=0.2489\linewidth,trim={0.5cm 0.15cm 0.5cm 0.1cm}, clip]{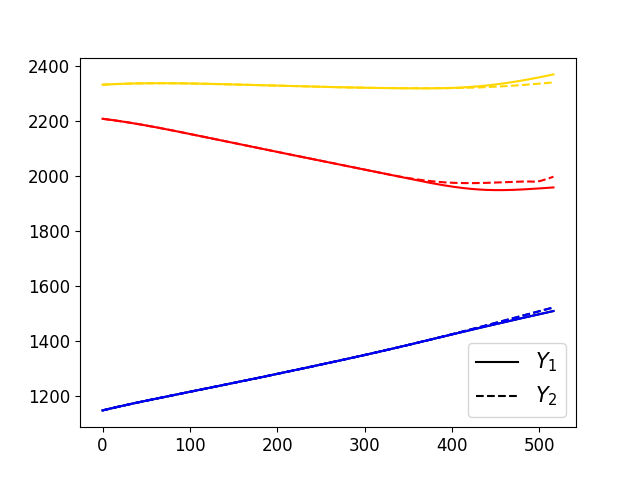}\label{f16_traj}} 
\subfloat[Spacecraft Trajectories]{
\includegraphics[scale=0.122]{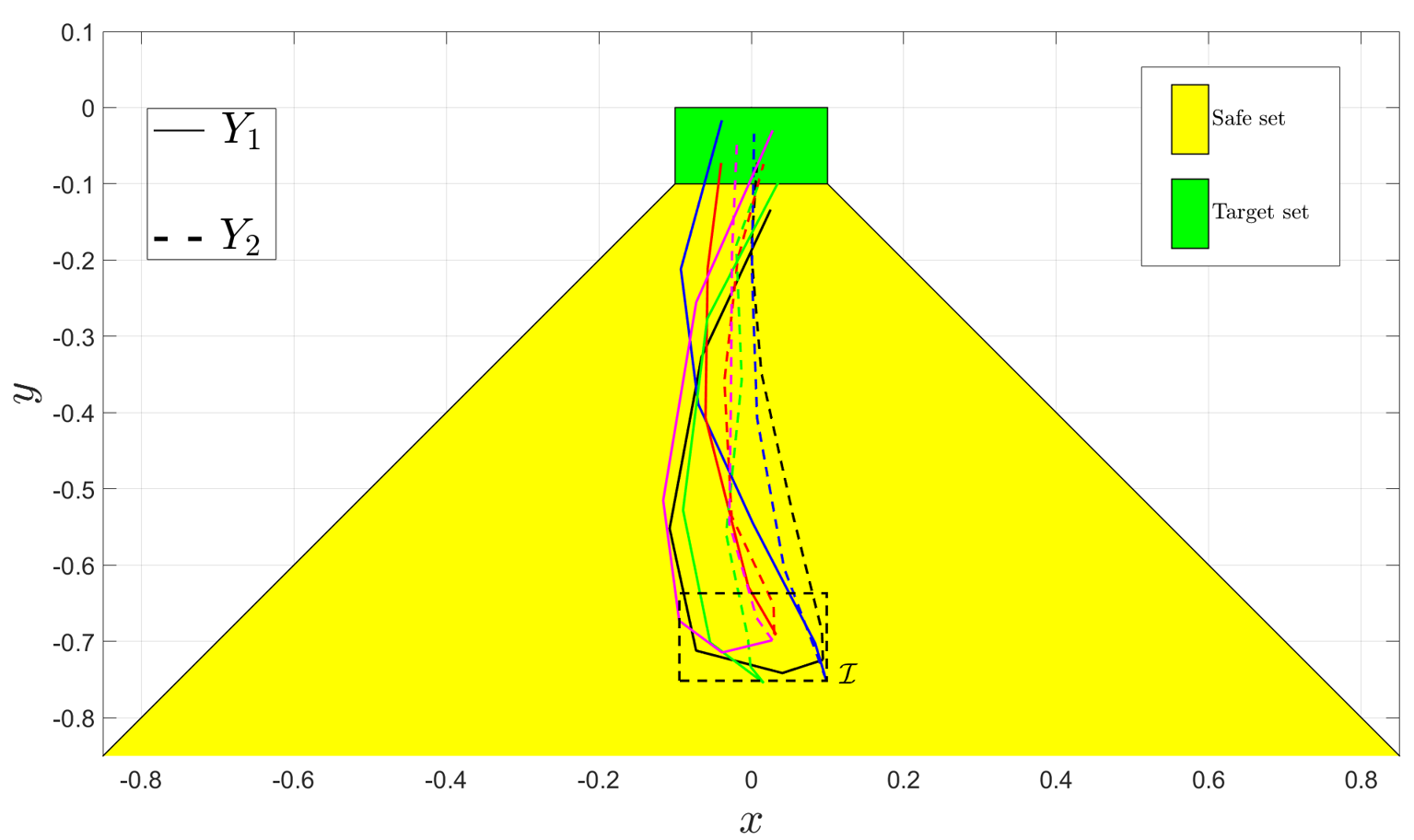}\label{subfig:spacecraft}} \hspace{-1mm}
\caption{The solid lines refer to $Y_1$ and the dashed lines refer to $Y_2$; in each of the displayed plots,
the initial condition for each pair of realizations
is the same.}\label{fig:jjkk}
\end{figure*}

We now demonstrate the practicality of stochastic conformance
and risk analysis through various case studies. For validation, if we obtain the value $\bar{Z}$ using a conformal prediction
procedure for a nonconformity score defined by the random variable
$Z$, i.e., such that $\text{Prob}(Z\le \bar{Z})\ge 1-\delta$. Then, given a test set  $D_{test}$, the validation score is defined as $VS(Z) := { |\{ z \in D_{test} \mid z \le \bar{Z} \}|}/
             { |D_{test}| }$.

\begin{figure*}[t]
\centering
\subfloat[Trajectory distance on validation set]{
\includegraphics[scale=0.31, trim={0.1cm 0cm 0.2cm 0.15cm}, clip]{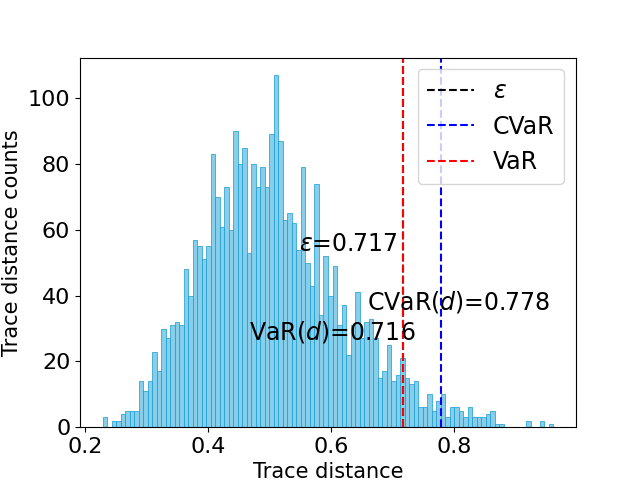}\label{subfig:q_dist_train}}
\quad
\subfloat[Robustness on validation set controller 1]{
\includegraphics[scale=0.31,trim={0.1cm 0cm 0.2cm 0.15cm}, clip]{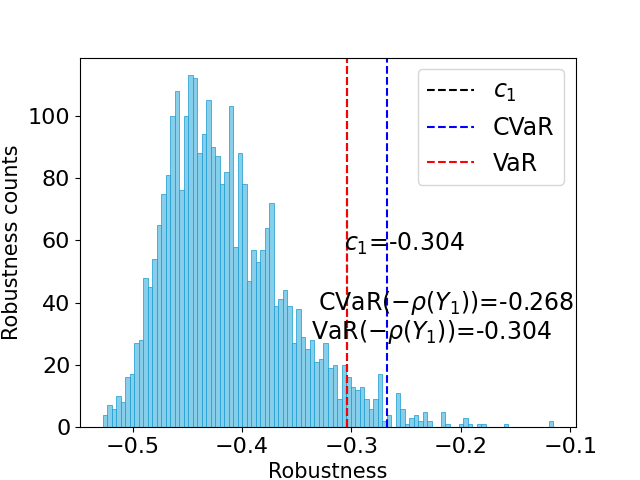}\label{subfig:q_rho_ctl1_train}}
\quad
\subfloat[Robustness on validation set controller 2]{
\includegraphics[scale=0.31,trim={0.1cm 0cm 0.2cm 0.3cm}, clip]{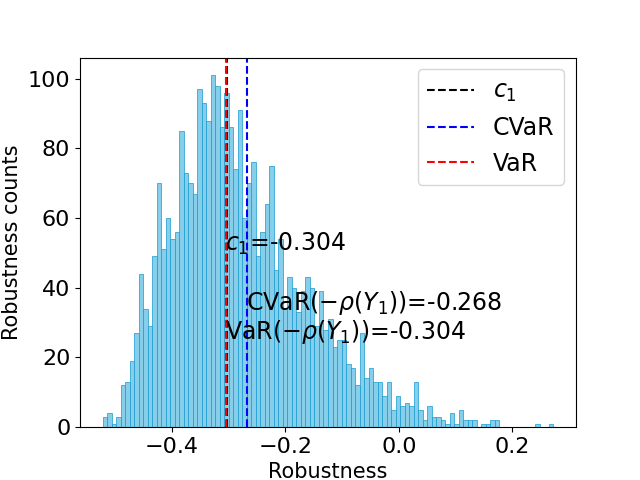}\label{subfig:q_rho_ctl2_train}}
\caption{Distance and robustness histogram for Dubin's car with $\delta=\bar{\delta}=0.05$. We use $CVaR(d)$  to denote $CVaR(d(Y_1, Y_2))$. The $c_1$ and $\epsilon$ are the values of conformal prediction on the calibration set of $\rho^{\phi_{dubin}}(Y_1)$ and $d_\infty(Y_1, Y_2)$.}\label{fig:dubin_hist}
\end{figure*}

\begin{table*}[t]
\centering
\begin{tabular*}{0.99\textwidth}{@{\extracolsep{\fill}}lllrrr}
\toprule
Distance & $|D_{cal}|$ & $\epsilon$ & \multicolumn{3}{c}{$d(Y_1,Y_2)$} \\
\cmidrule{4-6} 
Metric   &             &            & $VS(d(Y_1,Y_2))$ & $VaR(d(Y_1,Y_2))$ & $CVaR(d(Y_1,Y_2))$ \\
\midrule
\multirow{5}{2cm}{$d_\infty$}
    & 50   & 0.7825 & 0.987 & 0.7183 & 0.7947  \\ 
    & 1000 & 0.7163 & 0.956 & 0.7148 & 0.7647  \\  
    & 2000 & 0.7122 & 0.952& 0.712 &0.7814  \\  
    & 3000 & 0.7118 & 0.952& 0.7117 & 0.7862 \\
\midrule
\multirow{5}{2cm}{$d_{sk}$ (Skorokhod Distance)} 
    & 50   &  0.6723  & 0.953 & 0.6517 & 0.7181 \\ 
    & 1000 &  0.6722  & 0.972 & 0.6711 & 0.7156 \\  
    & 2000 &  0.6645  & 0.96  & 0.6639 & 0.7106 \\
    & 3000 &  0.6619  & 0.952 & 0.6613 & 0.7079 \\
\midrule
\multirow{5}{2cm}{$d_2$} 
    & 50   & 2.6086 & 0.937 & 2.503 & 2.612 \\ 
    & 1000 & 2.7339 & 0.96  & 2.732 & 3.048 \\ 
    & 2000 & 2.7071 & 0.944 & 2.706 & 3.044 \\ 
    & 3000 & 2.7238 & 0.955& 2.722 & 3.0929 \\
\bottomrule
\end{tabular*}
\caption{Effect of calibration set size on the validation score
and risk measures. The size of the test set, i.e., $|D_{test}|$, is 
1000. We use the conformal prediction procedure from
Section~\ref{sec:statest} to obtain $\epsilon$ as
defined in Definition~\ref{def:1} for $\delta = 0.05$. }
\label{tab:1}
\vspace{-2mm}

\end{table*}

\begin{table*}[t]
\centering
\begin{tabular*}{.99\textwidth}{@{\extracolsep{\fill}}lccccccccccc}
\toprule
& $|D_{cal}|$ & $c_1$ & $\epsilon$ & $VS(\rho_1)$ & $VS(d_\infty)$ & $c_2$ & Thm~\ref{thm:1} & 
\multicolumn{3}{c}{$CVaR$} & Thm.\ref{thm:2} \\
\cmidrule{9-11} 
&              &      &            &              &                &   &
valid? &
$-d_\infty$ & $-\rho_1$ & $-\rho_2$ & valid? \\
\midrule
$\delta = 0.2$, & 100& 0.31 & 0.59 & 0.95 & 0.76 & 0.21 & Y & 0.90 & -0.28 & 0.00 & Y \\
$\bar{\delta} = 0.05$ & 3K & 0.30 & 0.60 & 0.95 & 0.79 & 0.20 & Y & 0.93 & -0.27 & 0.03 & Y \\ 
\midrule
$\delta =0.1$,  & 1K & 0.30 &0.67  & 0.96 & 0.92 & 0.15 & Y & 0.79 & -0.27 & 0.02 & Y \\
$\bar{\delta} = 0.05$ & 3K & 0.30 & 0.66 & 0.95 & 0.91 & 0.15 & Y & 0.81 & -0.27 & 0.03 & Y \\ 
\midrule
$\delta =0.05$, & 2K & 0.31 & 0.71 & 0.94 & 0.95 & 0.11 & Y & 0.78 & -0.27 & 0.02 & Y \\ 
$\bar{\delta} = 0.05$ & 3K & 0.30 & 0.71 & 0.95 & 0.95 & 0.11 & Y & 0.79 & -0.27 & 0.03 & Y \\ 
\bottomrule
\end{tabular*}
\caption{Empirical evaluation of transference.
Let $\rho_i$ be short-hand for $\rho^{\varphi_{dubin}}(Y_i)$ for $i=1,2$,
and $d_\infty$ be short-hand for $d_\infty(Y_1,Y_2)$.
Using Theorem~\ref{thm:5}, we 
show $\text{Prob}(\rho_1 \ge c_1) > 1-\bar{\delta}$, and
$\text{Prob}(d_\infty \le \epsilon) > 1-\delta$. 
The validity scores for each guarantee
on a test set $D_{test}$ with 1000 samples are shown.
The value $c_2$ is obtained using Theorem~\ref{thm:5}
on $\rho_2$ and observe that it 
exceeds $c_1 -\epsilon$, validating Theorem~\ref{thm:1}. Similarly,
we report the $CVaR$ values for $-\rho_1$ and $d_\infty$, and $CVaR(-\rho_1) 
+ CVaR(d_\infty) \ge CVaR(-\rho_2)$ for all cases, validating 
Theorem~\ref{thm:2}. \label{tb:dubin_exp1}
}
\end{table*}

\begin{table*}[]
\centering
\begin{tabular*}{.99\textwidth}{@{\extracolsep{\fill}}lccc cccc}
\toprule
Case Study & Spec & $|D_{cal}|$ & $|D_{test}|$ & $VS(\rho_1)$ & $VS(d_\infty)$ & $\epsilon$ & $VaR(d_\infty)$ \\
\midrule
F-16      & $\phi_{gcas}$ & 1K & 3K & 0.95 & 0.98 & 200 & 200 \\
CARLA     & $\phi_{cte}$  & 700 & 300 & 0.94 & 0.96 & 1.88 & 1.87 \\
Satellite & $\phi_{sat}$  & 7K & 3K & 0.96 & 0.97 & 0.18 & 0.18 \\
\bottomrule
\end{tabular*}
\caption{Transference results for various case studies. We use 
$\delta = 0.05$ and $\bar{\delta} = 0.05$. As before, $\rho_1$ is
used as short-hand for $\rho^\phi(Y_1)$ for each spec, and
$d_\infty$ is used as short-hand for $d_\infty(Y_1,Y_2)$.
\label{tb:case_study_sum}}
\end{table*}

\newcommand{\mr}[1]{\multirow{1}{*}{#1}}
\begin{table*}[t]
\begin{tabular*}{0.99\textwidth}{@{\extracolsep{\fill}}lc cc c ccc}
\toprule
Case Study & Spec & $|D_{cal}|$ & $|D_{test}|$ & $\delta$ & \multicolumn{3}{c}{$CVaR$} \\
\cmidrule{6-8} 
           &      &             &              &          &         
$d_\infty$ & $-\rho_1$ & $-\rho_2$ \\
\midrule
\mr{F-16} & \mr{$\phi_{gcas}$} & \mr{1K} & \mr{3K} & 0.01 & 200.3 & -62.3 & -62.3 \\
\mr{CARLA}  & \mr{$\phi_{quad}$} & \mr{7K} & \mr{3K} & 0.01 & 2.04 & -0.31 & 0.88 \\
\mr{Satellite} & \mr{$\phi_{sat}$} & \mr{7K} & \mr{3K} & 0.01 & 0.19 & 0.0 & 0.08 \\ 
\bottomrule

\end{tabular*}
\caption{Empirical validation of risk transference for all case studies. As before, $\rho_i$ is short-hand for $\rho^\phi(Y_i)$, and $d_\infty$ is short-hand
for $d_\infty(Y_1,Y_2)$. Here, we set the risk level $\beta = \delta$ in each case.\label{tb:5}}
\end{table*}

\subsection{Dubin’s car.} Dubin's car models the motion of
a point mass vehicle. The state variables are the $x$ 
and $y$ position,  $\theta$ denotes the steering
angle and $v$ the longitudinal velocity. While both $\theta$ and $v$
are typically assumed to be control inputs, we adapt the case study from
\cite{vinod2019affine} where the angular velocity $\omega(t)$ at each time 
$t$ is assumed to be given so that 
$\theta(t) := T_s\pi + \sum_{i=1}^t \omega(i) T_s$ where $T_s := 0.1$s. In 
this example, we  assume that $\omega(i) := 
\frac{\pi}{50T_s}$ for $i \in [1,25]$, 
and $\omega(i) := -\frac{\pi}{50T_s}$ for $i\in[26,50]$. The velocity 
$v(t)$ is provided by a feedback controller. The dynamics 
are assumed to have additive white Gaussian noise  $\eta^x(t), \eta^y(t) \sim \mathcal{N}(0,0.005)$. The dynamical 
equations of motion are as described below:
\begin{align*}
x(t\!+\!1) = x(t)\!+\!T_s v(t) \cos(\theta(t)) + \eta^x(t) \\  
y(t\!+\!1) = y(t)\!+\!T_s v(t) \sin(\theta(t)) + \eta^y(t)
\end{align*}







The two systems that we compare have two different feedback controllers.
The first feedback controller uses the method from 
\cite{lesser2013stochastic,vinod2019affine} and the second controller uses the method from \cite{vitus2011feedback}.
We plot a set of sampled trajectories in Fig.~\ref{subfig:Dubin}. 
This figure also shows the set of initial states 
$\mathcal{I}:=[-1,\ 0] \times [-1,\ 0]$. The controller aims to
ensure that the system trajectory stays within a series of sets $\mathcal{T}_1$
through $\mathcal{T}_{50}$, the corresponding STL specification is
\(\phi_{dubin} := \bigwedge_{i=1}^{50} F_{[i-1,i]} \left([x_i\ y_i] \in \mathcal{T}_i\right)\).
For the experiments that follow, we uniformly sampled initial states
from $\mathcal{I}$ and noise $\eta^x,\eta^y$ from the described Gaussian distribution. 

\noindent{\em Effect of calibration set size.} In the first experiment,
we wish to benchmark the effect of the size of the calibration set
$D_{cal}$ for various distance metrics. The results are shown in  
Table~\ref{tab:1}. The table shows that with smaller sizes of 
the calibration set, we get a more conservative $\epsilon$ for $d_\infty$ (which 
translates into a higher validation score). The $VaR$ is almost 
identical to the value of $\epsilon$ at larger $D_{cal}$ sizes. We 
note that the $CVaR$ values change with the value of $VaR$. A similar 
trend can be observed the Skorokhod
distance and the $L_2$-metric.

\noindent{\em Empirical evaluation of transference.}
We empirically demonstrate that Theorem~\ref{thm:1} holds. We use
$C(Y) = \rho^{\phi_{dubin}}(Y)$, i.e., the robust semantics w.r.t.
the property $\phi_{dubin}$, and the $L_\infty$ signal metric $d_\infty$. The results are shown in 
Table~\ref{tb:dubin_exp1}. We can see that the predicted 
upper bound for the robustness of realizations of $Y_2$ w.r.t.
$\phi_{dubin}$ is negative ($c_1-\epsilon$), so it is not possible
to conclude that the second system satisfies $\phi_{dubin}$ with
probability greater than $1-\delta-\bar{\delta}$. However, we note that $c_2$ 
is indeed greater than the bound $(c_1-\epsilon)$. Similarly, we
show that Theorem~\ref{thm:2} is also empirically validated by
computing the $CVaR$ values for the first system and the risk
measure on $d_\infty(Y_1,Y_2)$. We show the empirical distributions
of $d_\infty(Y_1,Y_2)$, and $\rho^{\phi_{dubin}}(Y_i)$ for $i=1,2$ in Figure~\ref{fig:dubin_hist}.

%



\noindent {\em Empirical evaluation of Theorem~\ref{thm:1_}.} We next apply Algorithms~\ref{alg:overview} and \ref{alg:overview_} to this case study. We grid the initial set of
states evenly into 25 cells with a grid size of $\kappa=0.02$. 
We sample 650  trajectories on each cell to obtain their calibration 
sets. Algorithm~\ref{alg:overview_} gives $\bar{Z} = 0.7562$ and $L\kappa = 
0.0687$, giving $\bar{Z} + L\kappa = 0.8249$. We then evaluate
on two test sets of unseen initial conditions with $|D_{test}|$ =  1000, 2500. The success rate on the test sets are 0.9996  and 1.0, with the goal success rate being 0.9. The experiments demonstrate 
the effectiveness of Theorem \ref{th:supU}.

\subsection{F-16 aircraft.}
The F-16 aicraft control system from
\cite{heidlauf2018verification} uses a 13-dimensional non-linear plant model
based on a 6 d.o.f. airplane model, and its dynamics describe force equations, 
moments,  kinematics, and engine behavior. We alter the original system $S_1$ 
from \cite{heidlauf2018verification} to a modified version $S_2$ by changing 
the controller gains. We evaluate the performance of the two systems on the ground collision avoidance scenario with the specification 
$\phi_{gcas} := G_{[0,T]}(h \geq 1000)$ where $T$ is the mission time and
$h$ is the altitude. For data collection, we perform uniform sampling of the initial states. We assume that the 
x-center of gravity (xcg) of the aircraft is a stochastic parameter
with uniform distribution on $[0,0.8]$. We obtain a calibration set $D_{cal}$ of size 1000 by uniform sampling of the
initial states and the xcg parameter. We separately sample
3000 signals for $D_{test}$. The results of transference 
and risk estimates are shown in Table~\ref{tb:case_study_sum}.

%

\subsection{Autonomous Driving using the CARLA simulator.} CARLA is a
high-fidelity simulator for testing of autonomous driving systems 
\cite{dosovitskiy2017carla}. We consider two learning-based lane-keeping 
controllers from \cite{lindemann2022risk}, one being an imitation 
learning controller ($S_1$) and another being a learned barrier function 
controller ($S_2$). We obtain $1000$ trajectories from each controller 
during a $180$ degree left turn, and we use $700$ of them for calibration 
and $300$ for testing. In this data, the initial states $(c_e, \theta_e)$ 
are drawn uniformly from $[-1, 1] \times [-0.4, 0.4]$ where $c_e$ is 
the deviation from the center of the lane center (cross track error) and 
$\theta_e$ is the orientation error. The STL specification $\phi_{cte} := 
G(|c_e| \leq 2.25)$ restricts $|c_e|$ to be bounded by  $2.25$. The
results are shown in Table~\ref{tb:case_study_sum}.

\subsection{Spacecraft Rendezvous }
Next, we consider a spacecraft rendezvous problem from 
\cite{vinod2019sreachtools}. Here, a deputy spacecraft is to rendezvous with 
a master spacecraft while staying within a line-of-sight cone. The system is a
4D model $s = [x,\ y,\ v_x,\ v_y ]^\top$ where $x,y \in \mathbb{R}$ are the 
relative horizontal and vertical distances between the two spacecrafts and 
$v_x, v_y \in \mathbb{R}$ are the relative vertical and horizontal velocities.
There are two different feedback controllers, using the same control algorithms we 
used in Dubin's car example (i.e., the controllers from \cite{lesser2013stochastic,vinod2019affine} and \cite{vitus2011feedback}).  
The STL specification is a reach-avoid specification (visually
depicted in Fig.~\ref{subfig:spacecraft}), which requires
the system to always stay in the yellow region and eventually
reach the target rectangle $\mathcal{T}$ shown: $\phi_{sat} := G_{[1,5]}\left(  y,|y|,|v_x|,|v_y| \leq 
                 -|x|,y_{max},v_{x,max},v_{y,max}\right) \wedge F_{[1,5]} s \in \mathcal{T}$.
The set of initial states is $\mathcal{I} = [-0.1,\ 0.1]\times[-0.1,\ 0.1]$. The system is assumed to have additive Gaussian process
noise with zero mean and a diagonal covariance matrix with
variances $10^{-4},10^{-4},5\times 10^{-8},5\times10^{-8}$. We uniformly sample $100$ different initial states from $\mathcal{I}$ and $100$ noise values sampled from the noise distribution. We
divide the dataset into $D_{cal}$ and $D_{test}$ with sizes
7K and 3K respectively. The results are shown in Table~\ref{tb:case_study_sum}.

\noindent{\em Discussion on results for Transference across
case studies.} We omit the column for Table~\ref{tb:case_study_sum} 
that shows the proportion of $D_{test}$ of the realizations of
$Y_2$ for which the bound $c_1 - \epsilon$ exceeds $c_2$, where
the $c_i$'s are the conformal bounds on $\rho_i$'s. For all case
studies this ratio was either 1.0 or close to 1.0,
establishing the  empirical validity of Theorem~\ref{thm:1}.
We also observe that above results show that it
is feasible to use stochastic conformance in a control
improvization loop, where we want to change a system controller
(perhaps for optimizing a performance objective) while allowing
only some degradation on probabilistic safety guarantees.

\section{Related Work}


Conformance has found applications in cyber-physical system design \cite{jin2014benchmarks,araujo2018sound} as well as in drug testing and other applications \cite{dimitrova2020conformance,biewer2020conformance,RaynaDimitrova}. Our work is inspired by existing works for \textbf{conformance of deterministic systems} by which we mean that systems are non-stochastic, see \cite{conformance_survey,khakpour2015notions} for surveys.  The authors in \cite{abbas2014conformance,abbas2014formal,abbas2015test} considered conformance testing between hybrid system. To capture distance between hybrid system trajectories that may exhibit discontinuities, signal metrics were considered that simultaneously quantify distance in space and time,  resembling notions of system closeness in the hybrid systems literature  \cite{goebel2009hybrid,goebel2012hybrid}. For instance, \cite{abbas2014conformance} proposes $(T, J, (\tau, \epsilon))$-closeness where $\tau$ and $\epsilon$ capture both timing distortions and state value mismatches, respectively, and where $T$ and $J$ quantify limits on the total time and number of discontinuities, respectively. A stronger notion compared to $(T, J, (\tau, \epsilon))$-closeness was proposed in \cite{deshmukh2015quantifying} by using the Skoroghod metric. The benefit of \cite{deshmukh2015quantifying} over the other notion is that it preserves the timing structure. All these works derive transference results with respect to timed linear temporal logic or metric interval temporal logic specifications. 

\textbf{Conformance of stochastic systems} has been less studied. The authors in \cite{leemans2020stochastic} propose precision and recall conformance measures based on the notion of entropy of stochastic automata. The authors in \cite{leemans2019earth} use the Wasserstein distance to quantify distance between two stochastic systems, which is fundamentally different from our approach. (Bi)simulation relations for stochastic systems were studied in \cite{julius2006approximate,julius2009approximations,haesaert2020robust}. \textcolor{black}{Such techniques can define behavioral relations for systems \cite{abate, julius}, and they can be used to transfer verification results between systems \cite{zamani}. The authors in \cite{delgrange2023wasserstein} utilize such behavioral relations to verify RL policies between a concrete and an abstract system. We remark that bisimulations are difficult to compute, see e.g., \cite{GIRARD2011568}, unlike our approach.} Probably closest to our work is \cite{wang2021probabilistic}. However, in this paper conformance is task specific which allows two systems to be conformant w.r.t. a system specification even when the systems produce completely different trajectories. Additionally, we consider a worst-case notion of conformance where no information about the input that excites both stochastic systems is available. 


\section{Conclusion}
\label{sec:conclusion}
We studied conformance of stochastic dynamical systems. Particularly, we defined conformance between two stochastic systems as  probabilistic bounds over the distribution of distances between model trajectories. Additionally, we proposed the non-conformance risk to reason about the risk of stochastic systems not being conformant. We showed that both notions have the transference property, meaning that conformant systems satisfy similar system specifications. Lastly, we showed how stochastic conformance and the non-conformance risk can be estimated from data using statistical tools such as conformal prediction.


\section*{Acknowledgments}
The authors would like to thank the anonymous reviewers for their feedback. The National Science Foundation supported this work through the following grants: CAREER award (SHF-2048094), CNS-1932620, CNS-2039087, FMitF-1837131, CCF-SHF-1932620, the Airbus Institute for Engineering Research, and funding by Toyota R\&D and Siemens Corporate Research through the USC Center for Autonomy and AI.

\bibliographystyle{IEEEtran}
\bibliography{literature}

\appendix
\section{Semantics of Signal Temporal Logic}
\label{app:STL}
For a signal $y:\mathbb{T}\to\mathbb{R}^n$, the semantics of an STL formula $\phi$ that is imposed at time $t$, denoted by $(y,t)\models \phi$, can be recursively computed based on the structure of $\phi$ using the following rules:
	\begin{align*}
	(y,\tau)\models \text{True} & \hspace{0.5cm} \text{iff} \hspace{0.5cm} \text{True},\\
	(y,\tau)\models \mu & \hspace{0.5cm} \text{iff} \hspace{0.5cm} h(y(\tau))\ge 0,\\
	(y,\tau)\models \neg\phi & \hspace{0.5cm} \text{iff} \hspace{0.5cm} (y,\tau)\not\models \phi,\\
	(y,\tau)\models \phi' \wedge \phi'' & \hspace{0.5cm} \text{iff} \hspace{0.5cm} (y,\tau)\models\phi' \text{ and } (y,\tau)\models\phi'',\\
	(y,\tau)\models \phi' U_I \phi'' & \hspace{0.5cm} \text{iff} \hspace{0.5cm} \exists \tau''\in (\tau\oplus I)\cap \mathbb{T} \text{ s.t. } (y,\tau'')\models\phi'' \text{ and } \forall \tau'\in(\tau,\tau'')\cap \mathbb{T}, (y,\tau')\models\phi'.
	\end{align*}

The robust semantics $\rho^{\phi}(y,t)$ provide more information than the semantics $(y,t)\models \phi$, and indicate how robustly a specification is satisfied or violated. We first define the predicate robustness as \begin{align*}
	\text{dist}^{\mu}(y,t):=\inf_{y'\in \text{cl}(O^\mu)} \|y(t)-y'\|.
\end{align*}
where $O^\mu:=\{y\in\mathbb{R}^n|h(y)\ge 0\}$ denotes the set of all states that satisfy the predicate $\mu$, $\text{cl}(\cdot)$ denotes the closure of a set, and $\|\cdot\|$ denotes a vector norm. We now can recursively calculate $\rho^{\phi}(y,t)$ based on the structure of $\phi$ using the following rules:
\begin{align*}
	\rho^\text{True}(y,\tau)& := \infty,\\
	\rho^{\mu}(y,\tau)& := \begin{cases} \text{dist}^{\neg\mu}(y,\tau) &\text{if } y(\tau)\in O^\mu\\
	-\text{dist}^{\mu}(y,\tau) &\text{otherwise,}
	\end{cases} \\
	\rho^{\neg\phi}(y,\tau) &:= 	-\rho^{\phi}(y,\tau),\\
	\rho^{\phi' \wedge \phi''}(y,\tau) &:= 	\min(\rho^{\phi'}(y,\tau),\rho^{\phi''}(y,\tau)),\\
	\rho^{\phi' U_I \phi''}(y,\tau) &:=  \underset{\tau''\in (\tau\oplus I)\cap \mathbb{T}}{\text{sup}}  \Big(\min\big(\rho^{\phi''}(y,\tau''), \underset{\tau'\in (\tau,\tau'')\cap \mathbb{T}}{\text{inf}}\rho^{\phi'}(y,\tau') \big)\Big).
	\end{align*}



\section{Hölder Continuity of Robust Semantics $\rho^\phi$ w.r.t. $d_{\infty}$}
\label{sec:hölder}
\begin{proposition}\label{lemma:hölder} 
Let $\phi$ be an STL specification and let $d_\infty:\mathcal{Y}\times\mathcal{Y}\to\mathbb{R}$ be the $L_\infty$ signal metric. Then, it holds that the robust semantics $\rho^\phi$ are Hölder continuous in the first argument  w.r.t. $d_\infty$ with $H=\gamma=1$.
\end{proposition}
\begin{proof}
 Let $y_1,y_2\in\mathcal{Y}$ be two deterministic signals. We would like to show that, for a fixed time $\tau\in\mathbb{T}$, it holds that 
\begin{align*}
    |\rho^\phi(y_1,\tau)-\rho^\phi(y_2,\tau)|\le d_\infty(y_1,y_2).
\end{align*}
The proof idea is largely based on ideas from  \cite[Lemma 2]{cleaveland2022risk} in which it was, however, only stated that $\rho^\phi(y_1,\tau)-\rho^\phi(y_2,\tau)\le d_\infty(y_1,y_2)$. While the other direction follows trivially, we present the full proof using the notation of this paper for the convenience of the reader. Let us now assume that the STL formula $\phi$ is in positive normal form, e.g., that $\phi$ contains no negations. This assumption is made without loss of generality and the result holds for any STL formula since every STL formula $\phi$ can be rewritten into a semantically equivalent STL formula that is in positive normal form \cite[Proposition 2]{sadraddini2015robust}. We will perform the proof recursively on the structure of the formula $\phi$.

\textbf{Predicates. } First note that $\text{dist}^{\neg\mu}(y,\tau)$ is a Lipschitz continuous function with Lipschitz constant one in the sense that $|\text{dist}^{\mu}(y_1,\tau)-\text{dist}^{\mu}(y_2,\tau)|\le \|y_1(\tau)-y_2(\tau)\|$, see for instance \cite[Chapter 3]{munkres1975prentice}. Accordingly, it can easily be seen that $\rho^{\mu}(y,\tau)$ is Lipschitz continuous with Lipschitz constant one so that  $|\rho^{\mu}(y_1,\tau)-\rho^{\mu}(y_2,\tau)| \le \|y_1(\tau)-y_2(\tau)\|\le d_\infty(y_1,y_2)$.  

\textbf{Conjunctions. } For conjunctions $\phi'\wedge\phi''$ and by the induction assumption, it holds that $ |\rho^{\phi'}(y_1,\tau)- \rho^{\phi'}(y_2,\tau)| \le d_\infty(y_1,y_2)$ and $ |\rho^{\phi''}(y_1,\tau)- \rho^{\phi''}(y_2,\tau)| \le d_\infty(y_1,y_2)$. Now, it follows that
\begin{align*}
    \rho^{\phi'\wedge \phi''}(y_1,\tau)&=\min(\rho^{\phi'}(y_1,\tau),\rho^{\phi'}(y_1,\tau))\\
    &\ge \min(\rho^{\phi'}(y_2,\tau),\rho^{\phi''}(y_2,\tau))- d_\infty(y_1,y_2)\\
    &=\rho^{\phi'\wedge \phi''}(y_2,\tau)- d_\infty(y_1,y_2).
\end{align*}
Similarly, we can derive the other direction
\begin{align*}
    \rho^{\phi'\wedge \phi''}(y_1,\tau)&=\min(\rho^{\phi'}(y_1,\tau),\rho^{\phi'}(y_1,\tau))\\
    &\le \min(\rho^{\phi'}(y_2,\tau),\rho^{\phi''}(y_2,\tau))+ d_\infty(y_1,y_2)\\
    &=\rho^{\phi'\wedge \phi''}(y_2,\tau)+ d_\infty(y_1,y_2).
\end{align*}
Consequently, it follows that $|\rho^{\phi'\wedge \phi''}(y_1,\tau)-\rho^{\phi'\wedge \phi''}(y_2,\tau)|\le d_\infty(y_1,y_2)$.

\textbf{Until. } Using the induction assumption and the same reasoning as for conjunctions (it can be used in the same way for $\inf$ or $\sup$ operators), it follows that $|\rho^{\phi'U_I \phi''}(y_2,\tau)- \rho^{\phi'U_I \phi''}(y_2,\tau)|\le d_\infty(y_1,y_2)$.
\end{proof}


\section{H\"{o}lder Continuity of $C(y) = \rho^\varphi(y)$ w.r.t $d_{sk}$}
\label{sec:hölder2}
We recall from \cite{deshmukh2015quantifying} the definition of the
{\em Skorokhod} metric. In the following definition, $r:\mathbb{T}\to\mathbb{T}$
is a strictly increasing, bijective function known as a {\em retiming} function,
and $\mathcal{R}$ is the space of all possible retiming functions
\begin{equation}
\label{eq:skorokhod}
   d_{sk}(y_1,y_2) =  \inf_{r \in \mathcal{R}} 
                        \max\left( d_\infty(r(t),t), 
                                   d_\infty(y_1,y_2 \circ r) 
                            \right).
\end{equation}
For every retiming $r(t)$, the operand of the $\inf$ operator is 
the maximum between the magnitude of the retiming and the $d_\infty$
distance between $y_1$ and the retimed $y_2$. Skorokhod distance is 
then the infimum across all retimings.
\begin{proposition} 
Let $\phi$ be an STL specification and let $d_{sk}:\mathcal{Y}\times\mathcal{Y}\to\mathbb{R}$ be the Skorokhod metric. If the output signals are Lipschitz-continuous with
Lipschitz constant $K_y$, it holds that the robust semantics $\rho^\phi$ are Hölder continuous  w.r.t. $d_{sk}$ with $H=(1+K_y)$ and $\gamma=1$.
\end{proposition} 

\begin{proof}
Let $r(t)$ be the optimal retiming in the definition of the Skorokhod metric. Let
$d_{sk}(y_1,y_2) = \alpha$ for brevity. Under the optimal retiming $r$, this means
that $d_\infty(r(t),t)) \le \alpha$ and $d_\infty(y_1,y_2 \circ r) \le \alpha$. From the second term, we derive
\begin{align*}
&\alpha \ge d_\infty(y_1,y_2 \circ r) \\
& =  \sup_t \| y_1(t)  - y_2(r(t)) \| \\
& =  \sup_t \| (y_1(t) - y_2(t)) - (y_2(r(t)) - y_2(t)) \| \\
& \overset{(a)}{\ge}   \sup_t \| y_1(t) - y_2(t) \| - \sup_t \| y_2(r(t)) - y_2(t) \| \\
& \overset{(b)}{\ge}  d_\infty(y_1,y_2) - K_y \sup_t |r(t) - t| \\
& =  d_\infty(y_1,y_2) - K_y \alpha 
\end{align*}
where we used the reverse triangle inequality in $(a)$,
the fact that $y_2(t)$ is Lipschitz continuous with constant $K_y$ in $(b)$,
and that $-d_\infty(r(t),t) > -\alpha$ in $(c)$.
In other words, $\alpha \ge d_\infty(y_1,y_2) - K_y \alpha$, or
\begin{align*}
(1+K_y)d_sk(y_1,y_2) \ge d_\infty(y_1,y_2) \label{eq:sksup}
\end{align*}
From Proposition \ref{lemma:hölder}, we know that $|\rho^\phi(y_1)-\rho^\phi(y_2)| \le d_\infty(y_1,y_2)$ so that it follows that
\begin{align*}
|\rho^\phi(y_1) - \rho^\phi(y_2)| \le (1+K_y)d_{sk}(y_1,y_2).
\end{align*}

\end{proof}

\section{H\"{o}lder Continuity of $C(y) := \int_{0}^{T} y(t)^\top y(t)dt$ w.r.t $d_1$}
\label{sec:hölder3}

\begin{proposition}\label{prop:6}
Let $C(y) := \int_{0}^{T} y(t)^\top y(t)dt$ be the performance function and let $d_1:\mathcal{Y}\times\mathcal{Y}\to\mathbb{R}$ be the $L_1$ signal metric. If the output signal is of bounded magnitude, 
i.e., there exists $y_{\max}$ s.t. $\|y(t)\| < y_{\max}$, then it holds that the performance function $C$ is Hölder continuous w.r.t. $d_1$ with $H=2y_{\max}$ and $\gamma=1$.
\end{proposition}
\begin{proof}
We can derive the following chain of arguments
\begin{align*}
C(y_1) - C(y_2) & =  \int_{0}^{T} (y_1(t)^\top y_1(t) - y_2(t)^\top y_2(t))\ dt  \\
                & =  \int_{0}^T (\|y_1(t)\|^2_2 -
                                  \|y_2(t)\|^2_2) dt \\
                & =  \int_{0}^T (\|y_1(t)\| + \|y_2(t)\|)
                                 (\|y_1(t)\| - \|y_2(t)\|)dt\\
                & \overset{(a)}{\le} \int_{0}^T 2y_{\max}\|y_1(t) - y_2(t)\|dt    \\
                & \overset{(b)}{\le}  2y_{\max} d_1(y_1,y_2)
\end{align*}
where $(a)$ follows by the use of the reverse triangle inequality and $(b)$ follows by the boundedness assumption of signals. The other direction follows similarly so that the result follows trivially.
\end{proof}

\addtolength{\textheight}{-2cm}   

\end{document}